\documentclass[a4paper,UKenglish,cleveref,draft]{llncs}
\usepackage{xcolor}
\usepackage{soul}
\usepackage[utf8]{inputenc}
\usepackage{amsmath,amssymb}
\usepackage{xspace}
\usepackage{stmaryrd} 
\usepackage{todonotes}
\usepackage{enumitem}
\usepackage{algorithmicx}
    \usepackage[ruled]{algorithm}
\usepackage[noend]{algpseudocode}
\usepackage[final]{hyperref}

\newcommand{\assume}{\rho_0}
\newcommand{\target}{\rho_1}

\newcommand{\into}{\hookrightarrow}

\newcommand{\by}[1]{ && \text{(#1)}}

\newcommand{\glbox}{[\forall]}
\newcommand{\gldia}{[\exists]}

\renewcommand{\Box}{\square}

\newcommand{\Pow}{\mathcal{P}}

\newcommand{\clos}{\mathsf{cl}}
\newcommand{\dfr}{\mathsf{dfr}}
\newcommand{\clf}[1]{\theta^*(#1)}
\newcommand{\clo}[1]{\mathsf{cl}(#1)}

\newcommand{\sem}[1]{\llbracket #1 \rrbracket}

\newcommand{\infrule}[2]{\frac{#1}{#2}}

\spnewtheorem{thm}[theorem]{Theorem}{\bfseries}{\itshape}
\spnewtheorem{cor}[theorem]{Corollary}{\bfseries}{\itshape}
\spnewtheorem{cnj}[theorem]{Conjecture}{\bfseries}{\itshape}
\spnewtheorem{lem}[theorem]{Lemma}{\bfseries}{\itshape}
\spnewtheorem{lemdefn}[theorem]{Lemma and Definition}{\bfseries}{\itshape}
\spnewtheorem{prop}[theorem]{Proposition}{\bfseries}{\itshape}
\spnewtheorem{defn}[theorem]{Definition}{\bfseries}{\upshape}
\spnewtheorem{rem}[theorem]{Remark}{\bfseries}{\upshape}
\spnewtheorem{notation}[theorem]{Notation}{\bfseries}{\upshape}
\spnewtheorem{expl}[theorem]{Example}{\bfseries}{\upshape}
\spnewtheorem{thmdefn}[theorem]{Theorem and Definition}{\bfseries}{\itshape}
\spnewtheorem{propdefn}[theorem]{Proposition and Definition}{\bfseries}{\itshape}
\spnewtheorem{assumption}[theorem]{Assumption}{\bfseries}{\upshape}
\spnewtheorem{algo}[theorem]{Algorithm}{\bfseries}{\upshape}

 \renewenvironment{theorem}{\begin{thm}}{\end{thm}}

 \renewenvironment{definition}{\begin{defn}}{\end{defn}}
 \renewenvironment{remark}{\begin{rem}}{\end{rem}}
 \renewenvironment{example}{\begin{expl}}{\end{expl}}

\newcommand{\ALC}{\mathcal{ALC}}
\makeatletter
\def\moverlay{\mathpalette\mov@rlay}
\def\mov@rlay#1#2{\leavevmode\vtop{%
   \baselineskip\z@skip \lineskiplimit-\maxdimen
   \ialign{\hfil$\m@th#1##$\hfil\cr#2\crcr}}}
\newcommand{\charfusion}[3][\mathord]{
    #1{\ifx#1\mathop\vphantom{#2}\fi
        \mathpalette\mov@rlay{#2\cr#3}
      }
    \ifx#1\mathop\expandafter\displaylimits\fi}
\makeatother

\newcommand{\Sem}[1]{{[\![#1]\!]}}

\newcommand{\psem}[1]{\widehat{[\![#1]\!]}}

\newcommand{\diamod}[1]{{\langle #1\rangle}}
\newcommand{\boxmod}[1]{{[#1]}}

\usepackage[author=anonymous,margin,footnote,draft]{fixme}
\FXRegisterAuthor{dh}{adh}{DH}
\FXRegisterAuthor{ls}{als}{LS}

\DeclareMathOperator{\LFP}{\mathsf{LFP}}
\DeclareMathOperator{\GFP}{\mathsf{GFP}}

\DeclareMathOperator{\cl}{cl}

\newcommand\ExpTime{$\textsc{ExpTime}$\xspace}
\newcommand\PSpace{$\textsc{PSpace}$\xspace}
\newcommand\NP{$\textsc{NP}$\xspace}

\newcommand{\myparagraph}[1]{\medskip\par\noindent\textbf{\textbf{#1}}\hspace{6pt}}

\begin{document}

\title{NP Reasoning in the Monotone $\mu$-Calculus}

\author{Daniel Hausmann and Lutz Schr\"oder}

\institute{Friedrich-Alexander-Universit\"at Erlangen-N\"urnberg, Germany} 

\maketitle

\begin{abstract} Satisfiability checking for monotone modal logic is
  known to be (only) NP-complete. We show that this remains true when
  the logic is extended with alternation-free
  fixpoint operators as well as the universal modality; the resulting
  logic -- the \emph{alternation-free monotone
    $\mu$-calculus with the universal modality} -- contains both
  concurrent propositional dynamic logic (CPDL) and the
  alternation-free fragment of game logic as fragments. We obtain our
  result from a characterization of satisfiability by means of Büchi
  games with polynomially many Eloise nodes.
\end{abstract}


\section{Introduction}
\emph{Monotone modal logic} differs from normal modal logics (such
as~$K$~\cite{Chellas80}, equivalent to the standard description
logic~$\ALC$~\cite{BaaderEA03}) by giving up distribution of
conjunction over the box modality, but retaining monotonicity of the
modalities. Its semantics is based on (monotone) neighbourhood models
instead of Kripke models. Monotone modalities have been variously used
as epistemic operators that restrict the combination of knowledge by
epistemic agents~\cite{Vardi89}; as next-step modalities in the
evolution of concurrent systems, e.g.\ in \emph{concurrent
  propositional dynamic logic} (CPDL)~\cite{Peleg87}; and as game
modalities in systems where one transition step is determined by moves
of two players, e.g.\ in Parikh's \emph{game
  logic}~\cite{Parikh85,PaulyParikh03,HansenKupke04,EnqvistEA19}. The
monotonicity condition suffices to enable formation of fixpoints; one
thus obtains the \emph{monotone $\mu$-calculus}~\cite{EnqvistEA15},
which contains both CPDL and game logic as fragments (indeed, the
recent proof of completeness of game logic~\cite{EnqvistEA19} is based
on embedding game logic into the monotone $\mu$-calculus).

While many modal logics (including $K$/$\ALC$) have \PSpace-complete
satisfiability problems in the absence of fixpoints, it is known that
satisfiability in monotone modal logic is only
\NP-complete~\cite{Vardi89} (the lowest possible complexity given that
the logic has the full set of Boolean connectives). In the present
paper, we show that the low complexity is preserved under two
extensions that usually cause the complexity to rise from
\PSpace-complete to \ExpTime-complete: Adding the \emph{universal
  modality} (equivalently global axioms or, in description logic
parlance, a general TBox) and alternation-free
\emph{fixpoints}; that is, we show that satisfiability checking in the
\emph{alternation-free fragment of the monotone
  $\mu$-calculus with the universal modality}~\cite{EnqvistEA15} is
only \NP-complete. This logic subsumes both CPDL and the
alternation-free fragment of game logic~\cite{PaulyParikh03}. Thus,
our results imply that satisfiability checking in these logics is only
\NP-complete (the best previously known upper bound being~\ExpTime{}
in both cases~\cite{Parikh85,Peleg87,PaulyParikh03}); for comparison,
standard propositional dynamic logic (PDL) and in fact already the
extension of~$K$ with the universal modality are \ExpTime-hard. (Our
results thus seemingly contradict previous results on
\ExpTime-completeness of CPDL. However, these results rely on
embedding standard PDL into CPDL, which requires changing the
semantics of CPDL to interpret atomic programs as sequential programs,
i.e.\ as relations rather than neighbourhood systems~\cite{Peleg87}.)
Our results are based on a variation of the game-theoretic approach to
$\mu$-calculi~\cite{NiwinskiWalukiewicz96}. Specifically, we reduce
satisfiability checking to the computation of winning regions in a
satisfiability game that has exponentially many
$\mathsf{Abelard}$-nodes but only polynomially many
$\mathsf{Eloise}$-nodes, so that history-free winning strategies for
$\mathsf{Eloise}$ have polynomial size. From this approach we also
derive a polynomial model property.

\myparagraph{Organization} We recall basics on fixpoints and games in
Section~\ref{sec:prelim}, and the syntax and semantics of the monotone
$\mu$-calculus in Section~\ref{sec:monotone-mu}. We discuss a key
technical tool, \emph{formula tracking}, in
Section~\ref{sec:tracking}. We adapt the standard tableaux system to
the monotone $\mu$-calculus in Section~\ref{sec:tableaux}. In
Section~\ref{sec:games}, we establish our main results using a game
characterization of satisfiability.

\section{Notation and Preliminaries}\label{sec:prelim}

\noindent We fix basic concepts and notation on fixpoints and games.

\myparagraph{Fixpoints} Let $U$ be a set; we write $\Pow(U)$ for
the powerset of~$U$. Let $f:\Pow(U)\to\Pow(U)$ be a \emph{monotone}
function, i.e.\ $f(A)\subseteq f(B)$ whenever
$A\subseteq B\subseteq U$.  By the Knaster-Tarski fixpoint theorem,
the \emph{greatest} ($\GFP$) and \emph{least} ($\LFP$)
\emph{fixpoints} of~$f$ are given by
\begin{align*}
\LFP f =& \textstyle\bigcap\{V\subseteq U\mid f(V)\subseteq V\} &
\GFP f =& \textstyle\bigcup\{V\subseteq U\mid V\subseteq f(V)\},
\end{align*}
and are thus also the least prefixpoint ($f(V)\subseteq V$) and
greatest postfixpoint ($V\subseteq f(V)$) of~$f$, respectively. For
$V\subseteq U$ and $n\in\mathbb{N}$, we define $f^n(V)$ as expected by
$f^0(V)=V$ and $f^{n+1}(V)=f(f^n(V))$.  If $U$ is finite, then
$\LFP f = f^{|U|}(\emptyset)$ and $\GFP f = f^{|U|}(U)$ by Kleene's
fixpoint theorem.


\myparagraph{Infinite Words and Games} We denote the sets of finite
and infinite sequences of elements of a set $U$ by $U^*$ and
$U^\omega$, respectively.  We often view sequences
$\tau=u_1,u_2,\ldots$ in~$U^*$ (or $U^\omega$) as partial (total)
functions $\tau:\mathbb{N}\rightharpoonup U$, writing
$\tau(i)=u_i$. We write
\begin{align*}
\mathsf{Inf}(\tau)=\{u\in U\mid \forall i\in\mathbb{N}.\, \exists j>i.\,
\tau(j)=u\}
\end{align*}
for the set of elements of $U$ that occur infinitely often in $\tau$.

A \emph{B\"uchi game} $G=(V,E,v_0,F)$
consists of a set $V$ of \emph{nodes}, partitioned into the sets
$V_\exists$ of $\mathsf{Eloise}$-nodes and $V_\forall$ of
$\mathsf{Abelard}$-nodes, a set $E\subseteq V\times V$ of
\emph{moves}, an \emph{initial node} $v_0$, and a set $F\subseteq V$
of \emph{accepting} nodes. We write $E(v)=\{v'\mid (v,v')\in E\}$. For
simplicity, assume that $v_0\in V_\exists$ and that the game is
\emph{alternating}, i.e.\ $E(v)\subseteq V_{\forall}$ for all
$v\in V_\exists$, and $E(v)\subseteq V_{\exists}$ for all
$v\in V_\forall$ (our games will have this shape). A \emph{play} of
$G$ is a sequence $\tau=v_0,v_1,\ldots$ of nodes such that
$(v_i,v_{i+1})\in E$ for all $i\geq 0$ and~$\tau$ is either infinite
or ends in a node without outgoing moves. $\mathsf{Eloise}$ \emph{wins} a
play~$\tau$ if and only if $\tau$ is finite and ends in an
$\mathsf{Abelard}$-node or $\tau$ is infinite and
$\mathsf{Inf}(\tau)\cap F\neq\emptyset$, that is, $\tau$ infinitely
often visits an accepting node.  A \emph{history-free
  $\mathsf{Eloise}$-strategy} is a partial function
$s:V_\exists\rightharpoonup V$ such that $s(v_0)$ is defined, and
whenever $s(v)$ is defined, then $(v,s(v))\in E$ and $s(v')$ is
defined for all $v'\in E(s(v))$. A play~$v_0,v_1,\dots$ is an
\emph{$s$-play} if $v_{i+1}=s(v_i)$ whenever $v_i\in V_\exists$. We
say that~$s$ is a \emph{winning strategy} if~$\mathsf{Eloise}$ wins
every $s$-play, and that $\mathsf{Eloise}$ \emph{wins}~$G$ if there is
a winning strategy for $\mathsf{Eloise}$. B\"uchi games are
history-free determined, i.e.\ in every B\"uchi game, one of the
players has a history-free winning strategy~\cite{Mazala02}.

\section{The Monotone $\mu$-Calculus}\label{sec:monotone-mu}

\noindent We proceed to recall the syntax and semantics of the
monotone $\mu$-calculus.

\subsubsection*{Syntax} We fix countably infinite sets $\mathsf{P}$,
$\mathsf{A}$ and $\mathsf{V}$ of \emph{atoms}, \emph{atomic programs}
and \emph{(fixpoint) variables}, respectively; we assume that
$\mathsf{P}$ is closed under duals (i.e.\ atomic negation), i.e.\
$p\in \mathsf{P}$ implies $\overline{p}\in \mathsf{P}$, where
$\overline{\overline{p}}=p$.  Formulae of the \emph{monotone
  $\mu$-calculus} (in negation normal form) are then defined by the
grammar
\begin{equation*}
\psi,\phi:= \bot\mid\top \mid p \mid \psi\wedge\phi \mid \psi\vee\phi \mid \diamod{a}\psi \mid 
\boxmod{a}\psi \mid  X \mid \eta X.\psi 
\end{equation*}
where $p\in \mathsf{P}$, $a\in \mathsf{A}$, $X\in \mathsf{V}$; throughout,
we use $\eta\in\{\mu,\nu\}$ to denote extremal fixpoints. 
 As usual, $\mu$ and~$\nu$ are understood as
taking least and greatest fixpoints, respectively, and bind their
variables, giving rise to the standard notion of \emph{free variable}
in a formula~$\psi$. We write $\mathsf{FV}(\psi)$ for the set of free
variables in~$\psi$, and say that~$\psi$ is \emph{closed} if
$\mathsf{FV}(\psi)=\emptyset$.  Negation~$\neg$ is not included but can be
defined by taking negation normal forms as usual, with
$\neg p=\overline p$. We refer to formulae of the shape $\boxmod{a}\phi$ or $\diamod{a}\phi$ as \emph{($a$-)modal literals}. As indicated in the introduction, the
\emph{modalities} $\boxmod{a}$, $\diamod{a}$ have been equipped with
various readings, recalled in more detail in
Example~\ref{expl:game-cpdl}. 

Given a closed formula $\psi$,
the \emph{closure}  $\clo{\psi}$ of $\psi$ is
defined to be the least set of formulae that contains $\psi$ 
and satisfies the following closure properties:
\begin{align*}
\text{if } & \psi_1\wedge\psi_2\in\clo{\psi} \text{ or }
\psi_1\vee\psi_2\in\clo{\psi},\text{ then }
\{\psi_1,\psi_2\}\subseteq\clo{\psi},\\
\text{if } & \diamod{a}\psi_1\in\clo{\psi}  \text{ or }
\boxmod{a}\psi_1\in\clo{\psi},
\text{ then }
\psi_1\in\clo{\psi},\\
\text{if } & \eta X.\psi_1\in\clo{\psi}, \text{ then }
\psi_1[\eta X.\psi_1/X]\in\clo{\psi},
\end{align*}
where $\psi_1[\eta X.\psi_1/X]$ denotes the formula that is
obtained from $\psi_1$ by replacing every free occurrence of $X$ in
$\psi_1$ with $\eta X.\psi_1$.  Note that all formulae in $\clo{\psi}$
are closed. We define the \emph{size}~$|\psi|$ of~$\psi$ as
$|\psi|=|\clo{\psi}|$.  
A formula~$\psi$ is \emph{guarded} if whenever
$\eta X.\,\phi\in\clo{\psi}$, then all free occurrences of~$X$ in
$\phi$ are under the scope of at least one modal operator.  \emph{We
  generally restrict to guarded formulae}; see however
Remark~\ref{rem:sharing}. A closed formula~$\psi$ is \emph{clean} if
all fixpoint variables in~$\psi$ are bound by exactly one fixpoint
operator. Then $\theta(X)$ denotes \emph{the} subformula $\eta X.\phi$
that binds~$X$ in~$\psi$, and~$X$ is a least (greatest) fixpoint
variable if $\eta=\mu$ ($\eta=\nu$).  We define a partial
order~$\geq_\mu$ on the least fixpoint variables in~$\target$
and~$\assume$ by $X\geq_\mu Y$ iff $\theta(Y)$ is a subformula of
$\theta(X)$ and $\theta(Y)$ is not in the scope of a greatest fixpoint
operator within $\theta(X)$ (i.e.\ there is no greatest fixpoint
operator between $\mu X$ and $\mu Y$). The \emph{index}
$\mathsf{idx}(X)$ of such a fixpoint variable $X$ is
\begin{equation*}
\mathsf{idx}(X)=|\{Y\in\mathsf{V}\mid Y\geq_\mu X\}|,
\end{equation*}
For a subformula~$\phi$ of~$\psi$, we write
$\mathsf{idx}(\phi)=\max\{\mathsf{idx}(X)\mid
X\in\mathsf{FV}(\phi)\}$. We denote by $\clf{\phi_0}$ the closed
formula that is obtained from a subformula~$\phi_0$ of~$\psi$ by
repeatedly replacing free variables~$X$ with $\theta(X)$. Formally, we
define $\clf{\phi_0}$ as $\phi_{|\mathsf{FV}(\phi_0)|}$,
where~$\phi_{i+1}$ is defined inductively from~$\phi_i$.  If $\phi_i$
is closed, then put $\phi_{i+1}=\phi_i$. Otherwise, pick the
variable~$X_i\in\mathsf{FV}(\phi_i)$ with the greatest index and put
$\phi_{i+1}=\phi_i[\theta(X_i)/X_i]$.  Then
$\mathsf{idx}(\phi_{i+1})<\mathsf{idx}(\phi_i)$, so $\clf{\phi_0}$
really is closed; moreover, one can show that
$\clf{\phi_0}\in\clo{\psi}$.  A clean formula is
\emph{alternation-free} if none of its subformulae contains both a
free least and a free greatest fixpoint variable.
Finally,~$\psi$
is \emph{irredundant} if $X\in\mathsf{FV}(\phi)$ whenever
$\eta X.\phi\in\clo{\psi}$.

\begin{rem}\label{rem:sharing}
  We have defined the size of formulae as the cardinality of their
  closure, implying a very compact representation~\cite{BruseFL15}.
  Our upper complexity bounds thus become \emph{stronger}, i.e.\ they
  hold even for this small measure of input size. Moreover, the
  restriction to guarded formulae is then without loss of generality,
  since one has a \emph{guardedness transformation} that transforms
  formulae into equivalent guarded ones, with only polynomial blowup
  of the closure~\cite{BruseFL15}.
\end{rem}
\subsubsection*{Semantics} The monotone $\mu$-calculus is interpreted
over \emph{neighbourhood models} (or \emph{epistemic
  structures}~\cite{Vardi89}) $F=(W,N,I)$ where
$N:\mathsf{A}\times W\to 2^{(2^W)}$ assigns to each atomic program~$a$
and each state~$w$ a set $N(a,w)\subseteq 2^W$ of
\emph{$a$-neighbourhoods} of~$w$, and $I: \mathsf{P}\to 2^W$ interprets
propositional atoms such that $I(p)=W\setminus I(\overline{p})$ for $p\in\mathsf{P}$ (by $2$, we denote the set $\{\bot,\top\}$ of
Boolean truth values, and $2^W$ is the set of maps~$W\to 2$, which is
in bijection with the powerset $\Pow(W)$). Given such an~$F$, each
formula~$\psi$ is assigned an \emph{extension}
$\sem{\psi}_\sigma\subseteq W$ that additionally depends on a
\emph{valuation} $\sigma:V\to 2^W$, and is inductively defined by
\begin{align*}
\sem{p}_\sigma & = I(p) & \sem{X}_\sigma &= \sigma(X)\\
\sem{\psi\wedge\phi}_\sigma & = \sem{\psi}_\sigma\cap \sem{\phi}_\sigma&
\sem{\psi\vee\phi}_\sigma & = \sem{\psi}_\sigma\cup \sem{\phi}_\sigma\\
\sem{\diamod{a}\psi}_\sigma & = \{w\in W\mid \exists S\in N(a,w).\,S\subseteq
\sem{\psi}_\sigma\}&  
\sem{\mu X.\psi}_\sigma&=\LFP \sem{\psi}^X_{\sigma}\\
\sem{\boxmod{a}\psi}_\sigma & = \{w\in W\mid 
\forall S\in N(a,w).\,S\cap\sem{\psi}_\sigma\neq \emptyset\}& 
\sem{\nu X.\psi}_\sigma&=\GFP \sem{\psi}^X_{\sigma}
\end{align*}
where, for $U\subseteq W$ and fixpoint variables $X,Y\in \mathsf{V}$,
we put $\sem{\psi}^X_\sigma(U)=\sem{\psi}_{\sigma[X\mapsto U]}$,
$(\sigma[X\mapsto U])(X)= U$ and $(\sigma[X\mapsto U])(Y)= \sigma(Y)$
if $X\neq Y$. We omit the dependence on~$F$ in the notation
$\Sem{\phi}_\sigma$, and when necessary clarify the underlying neighbourhood
model by phrases such as `in~$F$'. If $\psi$ is closed, then its extension
does not depend on the valuation, so we just write $\sem{\psi}$.
A closed formula $\psi$ is \emph{satisfiable} if there is a
neighbourhood model~$F$ such that $\sem{\psi}\neq\emptyset$ in~$F$; in
this case, we also say that~$\psi$ is \emph{satisfiable over~$F$}.
Given a set $\Psi$ of closed formulae, we write
$\sem{\Psi}=\bigcap_{\psi\in\Psi}\sem{\psi}$.
An \emph{(infinite) path} through a neighbourhood model $(W,N,I)$ is a sequence $x_0,x_1,\ldots$ of states $x_i\in W$ such that
for all $i\geq 0$, there are $a\in\mathsf{A}$ and $S\in N(a,x_i)$ such that
$x_{i+1}\in S$.

The soundness direction of our game characterization will rely on the
following immediate property of the semantics, which may be seen as
soundness of a modal tableau rule~\cite{CirsteaEA11a}.

\begin{lem}~\cite[Proposition~3.8]{Vardi89} \label{lem:rule} If\/
  $\boxmod{a}\phi\land\diamod{a}\psi$ is satisfiable over a
  neighbourhood model~$F$, then $\phi\land\psi$ is also satisfiable
  over~$F$.
\end{lem}
\begin{rem}\label{rem:box-vs-diamond}
  The dual box and diamond operators~$\boxmod{a}$ and $\diamod{a}$ are
  completely symmetric, and indeed the notation is not uniform in the
  literature. Our use of $\boxmod{a}$ and $\diamod{a}$ is generally in
  agreement with work on game logic~\cite{Parikh85} and
  CPDL~\cite{Peleg87}; in work on monotone modal logics and the
  monotone $\mu$-calculus, the roles of box and diamond are often
  interchanged~\cite{Vardi89,EnqvistEA15,BenthemEA19}.
\end{rem}
\begin{rem}
  The semantics may equivalently be presented in terms of
  \emph{monotone} neighbourhood models, where the set of
  $a$-neighbourhoods of a state is required to be upwards closed under
  subset
  inclusion~\cite{Parikh85,PaulyParikh03,HansenKupke04,EnqvistEA15}. In
  this semantics, the interpretation of $\diamod{a}\phi$ simplifies to
  just requiring that the extension of~$\phi$ is an $a$-neighbourhood
  of the current state. We opt for the variant where upwards closure
  is instead incorporated into the interpretation of the modalities,
  so as to avoid having to distinguish between monotone neighbourhood
  models and their representation as upwards closures of (plain)
  neighbourhood models, e.g.\ in small model theorems.
\end{rem}
We further extend the expressiveness of the logic (see also
Remark~\ref{rem:submodels}) by adding \emph{global assumptions} or
equivalently the \emph{universal modality}:
\begin{defn}[Global assumptions]
  Given a closed formula~$\phi$, a \emph{$\phi$-model} is a neighbourhood
  model $F=(W,N,I)$ in which $\Sem{\phi}=W$. A formula~$\psi$ is
  \emph{$\phi$-satisfiable} if~$\psi$ is satisfiable over some
  $\phi$-model; in this context, we refer to~$\phi$ as the
  \emph{global assumption}, and to the problem of deciding whether
  $\psi$ is $\phi$-satisfiable as \emph{satisfiability checking under
    global assumptions}.
\end{defn}
We also define an extension of the monotone $\mu$-calculus, the
\emph{monotone $\mu$-calculus with the universal modality}, by adding
two alternatives
\begin{equation*}
  \dots \mid \glbox \phi \mid \gldia \phi
\end{equation*}
to the grammar, in both alternatives restricting~$\phi$ to be
closed. The definition of the semantics over a neighbourhood model
$(W,N,I)$ and valuation~$\sigma$ is correspondingly extended by
$\Sem{\glbox\phi}_\sigma= W$ if $\Sem{\phi}_\sigma=W$, and
$\Sem{\glbox\phi}_\sigma= \emptyset$ otherwise; dually,
$\Sem{\gldia \phi}_\sigma= W$ if $\Sem{\phi}_\sigma\neq\emptyset$, and
$\Sem{\gldia\phi}_\sigma= \emptyset$ otherwise. That is, $\glbox\phi$ says
that~$\phi$ holds in all states of the model, and $\gldia\phi$
that~$\phi$ holds in some state.
\begin{rem}\label{rem:global}
  In description logic, global assumptions are typically called
  (general) \emph{TBoxes} or
  \emph{terminologies}~\cite{BaaderEA03}. For many next-step modal
  logics (i.e.\ modal logics without fixpoint operators),
  satisfiability checking becomes harder under global assumptions. A
  typical case is the standard modal logic~$K$ (corresponding to the
  description logic~$\mathcal{ALC}$), in which (plain) satisfiability
  checking is \PSpace-complete~\cite{Ladner77} while satisfiability
  checking under global assumptions is
  \ExpTime-complete~\cite{FischerLadner79,Donini03}. Our results show
  that such an increase in complexity does not happen for
  monotone modalities.

  For purposes of satisfiability checking, the universal modality and
  global assumptions are mutually reducible in a standard manner,
  where the non-trivial direction (from the universal modality to
  global assumptions) is by guessing beforehand which subformulae
  $\glbox\phi$, $\gldia\phi$ hold (see also~\cite{GorankoPassy92}).
\end{rem}

\begin{expl}\label{expl:game-cpdl}
\begin{enumerate}[wide]
\item In \emph{epistemic logic}, neighbourhood models have been termed
  \emph{epistemic structures}~\cite{Vardi89}. In this context, the
  $a\in \mathsf{A}$ are thought of as \emph{agents}, the
   $a$-neighbourhoods of
  a state~$w$ are the facts known to agent~$a$ in~$w$, and
  correspondingly the reading of $\boxmod{a}\phi$ is `$a$
  knows~$\phi$'. The use of (monotone) neighbourhood models and the
  ensuing failure of normality imply that agents can still weaken
  facts that they know but are not in general able to combine them,
  i.e.\ knowing~$\phi$ and knowing~$\psi$ does not entail knowing
  $\phi\land\psi$~\cite{Vardi89}.
\item In \emph{concurrent propositional dynamic logic (CPDL)},
  $a$-neighbourhoods of a state are understood as sets of states that
  can be reached concurrently, while the choice between several
  $a$-neighbourhoods of a state models sequential non-determinism.
  CPDL indexes modalities over composite \emph{programs} $\alpha$,
  formed using
  tests $?\phi$ and the standard operations of propositional dynamic
  logic (PDL) (union~$\cup$, sequential composition~`$;$', and
  Kleene star~$(-)^*$) and additionally intersection~$\cap$. 
  It forms a sublogic of Parikh's game logic, recalled next, and thus
  in particular translates into the monotone $\mu$-calculus.

  As indicated in the introduction, CPDL satisfiability checking has
  been shown to be \ExpTime-complete~\cite{Peleg87}, seemingly
  contradicting our results (Corollary~\ref{cor:game-cpdl}). Note
  however that the interpretation of atomic programs in CPDL,
  originally defined in terms of neighbourhood
  systems~\cite[p.~453]{Peleg87}, is, for purposes of the
  \ExpTime-hardness proof, explicitly changed to
  relations~\cite[pp.~458--459]{Peleg87}; \ExpTime-hardness then
  immediately follows since PDL becomes a sublogic of
  CPDL~\cite[p.~461]{Peleg87}. Our \NP bound applies to the original
  semantics. 
  %
\item \emph{Game logic}~\cite{Parikh85,PaulyParikh03} extends CPDL by
  a further operator on programs, \emph{dualization} $(-)^d$, and
  reinterprets programs as games between two players \emph{Angel} and
  \emph{Demon}; in this view, dualization just corresponds to swapping
  the roles of the players. In comparison to CPDL, the main effect of
  dualization is that one obtains an additional \emph{demonic
    iteration operator} $(-)^\times$, distinguished from standard
  iteration $(-)^*$ by letting Demon choose whether or not to continue
  the iteration. A game logic formula is \emph{alternation-free} if it
  does not contain nested occurrences (unless separated by a test) of
  $(-)^\times$ within~$(-)^*$ or vice versa~\cite{PaulyParikh03}.

  Enqvist et al.~\cite{EnqvistEA19} give a translation of game logic
  into the monotone $\mu$-calculus that is quite similar to
  Pratt's~\cite{Pratt81} translation of PDL into the standard
  $\mu$-calculus. The translation $(-)^\sharp$ is defined by
  commutation with all Boolean connectives and by
  \begin{equation*}
    (\diamod{\gamma}\phi)^\sharp = \tau_\gamma(\phi^\sharp),
  \end{equation*}
  in mutual recursion with a function $\tau_\gamma$ that
  translates the effect of applying $\diamod{\gamma}$ into the
  monotone $\mu$-calculus. (Boxes $\boxmod{\gamma}$ can be replaced
  with $\diamod{\gamma^d}$). We refrain from repeating the full
  definition of~$\tau_\gamma$ by recursion over~$\gamma$; some
  key clauses are
  \begin{equation*}
    \tau_{\gamma\cap\delta}(\psi)  = \tau_\gamma(\psi)\land\tau_\delta(\psi)\quad
    \tau_{\gamma^*}(\psi)  = \mu X.\,(\psi\lor\tau_\gamma(X))\quad
    \tau_{\gamma^\times}(\psi)  = \nu Y.\,(\psi\land\tau_\gamma(Y))
  \end{equation*}
  where in the last two clauses,~$X$ and~$Y$ are chosen as fresh
  variables (for readability, we gloss over a more precise treatment
  of this point given in~\cite{EnqvistEA19}). The first clause (and a
  similar one for~$\cup$) appear at first sight to cause exponential
  blowup but recall that we measure the size of formulae by the
  cardinality of their closure; in this measure, there is in fact no
  blowup. The translated formula $\psi^\sharp$ need not be guarded as
  the clauses for~$^*$ and~$^\times$ can introduce unguarded fixpoint
  variables; as mentioned in Remark~\ref{rem:sharing}, we can however
  apply the guardedness transformation, with only quadratic blowup of
  the closure~\cite{BruseFL15}.
  
  Under this translation, the alternation-free fragment of game logic ends
  up in the (guarded) alternation-free fragment of the monotone
  $\mu$-calculus.
\end{enumerate}
\end{expl}
\noindent For later use, we note
\begin{lem}\label{lem:fin-model}
  The monotone $\mu$-calculus with the universal modality has the
  finite model property.
\end{lem}
\begin{proof}[sketch]
  We reduce to global assumptions as per Remark~\ref{rem:global}, and
  proceed by straightforward adaptation of the translations of
  monotone modal logic~\cite{KrachtWolter99} and game
  logic~\cite{PaulyParikh03} into the \emph{relational}
  $\mu$-calculus, thus inheriting the finite model
  property~\cite{BradfieldStirling06}. This translation is based on
  turning neighbourhoods into additional states, connected to their
  elements via a fresh relation~$e$. Then, e.g., the monotone
  modality~$\boxmod{a}$ (in our notation, cf.\
  Remark~\ref{rem:box-vs-diamond}) is translated into
  $\boxmod{a}\diamod{e}$ (relational modalities).  Moreover, we
  translate a global assumption~$\phi$ into a formula saying that all
  reachable states satisfy~$\phi$, expressed in the $\mu$-calculus in
  a standard fashion. \qed
\end{proof}

\begin{rem}\label{rem:submodels}
  In the relational $\mu$-calculus, we can encode a
  modality~$\boxtimes$ `in all reachable states', generalizating the
  $\mathsf{AG}$ operator from CTL. As already indicated in the proof
  of Lemma~\ref{lem:fin-model}, this modality allows for a
  straightforward reduction of satisfiability under global assumptions
  to plain satisfiablity in the relational $\mu$-calculus: A
  formula~$\psi$ is satisfiable under the global assumption~$\phi$ iff
  $\psi\land\boxtimes\phi$ is satisfiable, where `if' is shown by
  restricting the model to reachable states. To motivate separate
  consideration of the universal modality, we briefly argue why an
  analogous reduction does not work in the monotone $\mu$-calculus.

  It is not immediately clear what reachability would mean on
  neighbourhood models. We can however equivalently rephrase the
  definition of~$\boxtimes$ in the relational case to let
  $\boxtimes\phi$ mean `the present state is contained in a submodel
  in which every state satisfies~$\phi$', where as usual a submodel of
  a relational model~$C$ with state set~$W$ is a model $C'$ with state
  set $W'\subseteq W$ such that the graph of the inclusion $W'\into W$
  is a bisimulation from~$C'$ to~$C$. We thus refer to~$\boxtimes$ as
  the \emph{submodel modality}. This notion transfers to neighbourhood
  models using a standard notion of bisimulation: A \emph{monotone
    bisimulation}~\cite{Pauly99,PaulyThesis} between neighbourhood
  models $(W_1,N_1,I_1)$, $(W_2,N_2,I_2)$ is a relation
  $S\subseteq W_1\times W_2$ such that whenever $(x,y)\in S$, then
  \begin{itemize}
  \item[--] for all $a\in\mathsf{A}$ and $A\in N_1(a,x)$,
    there is $B\in N_2(a,y)$ such that for all $v\in B$, there is
    $u\in A$ such that $(u,v)\in S$,
  \item[--] for all $a\in\mathsf{A}$ and $B\in N_2(a,y)$,
    there is $A\in N_1(a,x)$ such that for all $u\in A$, there
    is a $v\in B$ such that $(u,v)\in S$,
  \item[--] for all $p\in\mathsf{P}$, $x\in I_1(p)$ if and only if
    $y\in I_2(p)$.
  \end{itemize}
  We then define a \emph{submodel} of a neighbourhood model
  $F=(W,N,I)$ to be a neigbourhood model $F'=(W',N',I')$ such that
  $W'\subseteq W$ and the graph of the inclusion $W'\into W$ is a
  monotone bisimulation between~$F'$ and~$F$. If $(x,y)\in S$ for some
  monotone bisimulation~$S$, then~$x$ and~$y$ satisfy the same
  formulae in the monotone $\mu$-calculus~\cite{EnqvistEA15}. It
  follows that the submodel modality~$\boxtimes$ on neighbourhood
  models, defined verbatim as in the relational case, allows for the
  same reduction of satisfiability under global assumptions as the
  relational submodel modality.

  However, the submodel modality fails to be expressible in the
  monotone $\mu$-calculus, as seen by the following example. Let
  $F_1=(W_1,N_1,I_1)$, $F_2=(W_2,N_2,I_2)$ be the neighbourhood models
  given by $W_1=\{x_1,u_1,v_{11},v_{12}\}$, $W_2=\{x_2,u_2,v_2\}$,
  $N_1(a,x_1)=\{\{u_1,v_{11}\},\{v_{11},v_{12}\}\}$,
  $N_2(a,x_2)=\{\{v_{2}\}\}$, $N_1(b,y)=N_2(b,z)=\emptyset$ for
  $(b,y)\neq(a,x_1)$, $(b,z)\neq(a,x_2)$,
  $I_1(p)=\{x_1,v_{11},v_{12}\}$, $I_2(p)=\{x_2,v_2\}$, and
  $I_1(q)=I_2(q)=\emptyset$ for $q\neq p$. Then
  $S=\{(x_1,x_2),(u_1,u_2),(v_{11},v_2),(v_{12},v_2)\}$ is a monotone
  bisimulation, so~$x_1$ and~$x_2$ satisfy the same formulae in the
  monotone $\mu$-calculus. However,~$x_2$ satisfies $\boxtimes p$
  because~$x_2$ is contained in a submodel with set $\{x_2,v_2\}$ of
  states, while~$x_1$ does not satisfy~$\boxtimes p$.
\end{rem}

\section{Formula Tracking}\label{sec:tracking}

A basic problem in the construction of models in fixpoint
logics is to avoid infinite unfolding of least fixpoints, also known
as infinite deferral. To this end, we use a \emph{tracking function}
to follow formulae along paths in prospective models.
For unrestricted
$\mu$-calculi, infinite unfolding of least
fixpoints is typically detected by means of a \emph{parity condition}
on tracked formulae. However, since we restrict to alternation-free formulae, we can instead use a \emph{Büchi condition} to
detect (and then reject) such infinite unfoldings by sequential \emph{focussing} on sets of formulae in the spirit of focus games~\cite{LangeStirling01}. 

We fix \emph{closed, clean, irredundant formulae $\target$
and $\assume$}, aiming to check $\assume$-satisfiability of
$\target$; we also require both $\target$ and $\assume$ 
to be alternation-free. 
We put $\clos=\clo{\target}\cup\clo{\assume}$ and $n=|\clos|$.

Next we formalize our notion of deferred formulae
that originate from least fixpoints;
these are the formulae for which infinite unfolding has to be avoided.
\begin{defn}[Deferrals]
  A formula $\phi\in\clos(\psi)$ is a \emph{deferral} if there is a
  subformula $\chi$ of $\psi$ such that $\mathsf{FV}(\chi)$ contains a
  least fixpoint variable and $\phi=\clf{\chi}$. We put
  $\mathsf{dfr}= \{\phi\in\clos\mid \phi \text{ is a deferral}\}$.
\end{defn}
Since we assume formulas to be alternation-free, no formula $\nu X.\,\phi$
is a deferral.
\begin{example}
For $\psi=\mu X.\, \boxmod{b}\top\vee\diamod{a}X$, the formula
$\boxmod{b}\top\vee\diamod{a}\psi=(\boxmod{b}\top\vee\diamod{a}X)[\psi/X]\in\clo{\psi}$ is
a deferral since $X$ is a least fixpoint variable and
occurs free in $\boxmod{b}\top\vee\diamod{a}X$; the formula $\boxmod{b}\top$ on the
other hand is not a deferral since it cannot be obtained
from a subformula of $\psi$ by replacing least fixpoint variables with their binding fixpoint formulae.
\end{example}

\noindent We proceed to define a tracking function that
nondeterministically tracks deferrals along paths in neighbourhood
models.

\begin{defn}[Tracking function]
  We define an alphabet $\Sigma=\Sigma_p\cup\Sigma_m$ for traversing
  the closure, separating propositional and modal traversal, by
\begin{align*}
\Sigma_p&=\{(\phi_0\wedge\phi_1,0),(\phi_0\vee\phi_1,b),(\eta X.\phi_0,0)\mid\\
&\qquad\qquad\qquad\qquad\qquad\qquad\qquad
\phi_0,\phi_1\in\clos,\eta\in\{\mu,\nu\},X\in \mathsf{V},b\in\{0,1\}\}\\
\Sigma_m&=\{(\diamod{a}\phi_0,\boxmod{a}\phi_1)\mid
\phi_0,\phi_1\in\clos,a\in \mathsf{A}\}.
\end{align*}
The \emph{tracking function} $\delta:\dfr\times\Sigma\to \Pow(\dfr)$
is defined by $\delta(\mathsf{foc},(\chi,b))=\{\mathsf{foc}\}$ for
$\mathsf{foc}\in \dfr$ and $(\chi,b)\in\Sigma_p$ with
$\chi\neq\mathsf{foc}$, and by
\begin{align*}
\delta(\mathsf{foc},(\mathsf{foc},b))&=
   \begin{cases}
     \{\phi_0,\phi_1\}\cap\dfr &   \mathsf{foc}=\phi_0\wedge\phi_1\\ 
     \{\phi_b\}\cap\dfr & \mathsf{foc}=\phi_0\vee\phi_1 \\
     \{\phi_0[\theta(X)/X]\} & \mathsf{foc}=\mu X.\,\phi_0
   \end{cases}\\
\delta(\mathsf{foc},(\diamod{a}\phi_0,\boxmod{a}\phi_1))&=
   \begin{cases}
     \{\phi_0\}\hspace{4.3em} & \mathsf{foc}=\diamod{a}\phi_0\\
     \{\phi_1\} & \mathsf{foc}=\boxmod{a}\phi_1,
   \end{cases}
\end{align*}
noting that $\mu X.\,\phi_0\in\dfr$ implies
$\phi_0[\theta(X)/X]\in\dfr$.
We extend $\delta$ to sets $\mathsf{Foc}\subseteq \dfr$ by putting 
$\delta(\mathsf{Foc},l)=\bigcup_{\mathsf{foc}\in\mathsf{Foc}}
\delta(\mathsf{foc},l)$ for $l\in\Sigma$. We also
extend $\delta$ to words in the obvious way; e.g.\
we have $\delta(\mathsf{foc},w)=[a]\mathsf{foc}$ for
\begin{equation*}
\mathsf{foc}=\mu X.\,(p\vee(\langle b\rangle p\wedge[a]X))  
\quad w=(\mathsf{foc},0),(p\vee(\langle b\rangle p\wedge[a]\mathsf{foc}),1),
(\langle b\rangle p\wedge[a]\mathsf{foc},0).
\end{equation*}

\end{defn}

\begin{remark}
  Formula tracking can be modularized in an elegant way by using
  tracking automata to accept exactly the paths that contain infinite
  deferral of some least fixpoint formula.  While tracking automata
  for the full $\mu$-calculus are in general nondeterministic parity
  automata~\cite{FriedmannLange13a}, the automata for alternation-free
  formulae are nondeterministic co-B\"uchi automata.  Indeed, our
  tracking function~$\delta$ can be seen as the transition function of
  a nondeterministic co-B\"uchi automaton with set $\dfr$ of states
  and alphabet $\Sigma$ in which all states are accepting.  Since
  models have to be constructed from paths that do \emph{not} contain
  infinite deferral of some least fixpoint formula, the tracking
  automata have to be complemented in an intermediate step (which is
  complicated by nondeterminism) when checking satisfiability.  In our
  setting, the complemented automata are deterministic B\"uchi
  automata obtained using a variant of the Miyano-Hayashi
  construction~\cite{MiyanoHayashi1984} similarly as in our previous
  work on \ExpTime satisfiability checking in alternation-free
  coalgebraic $\mu$-calculi~\cite{HausmannEA16}; for brevity, this
  determinization procedure remains implicit in our satisfiability
  games (see Definition~\ref{defn:satgames} below).
\end{remark}

\noindent To establish our game characterization of satisfiability, 
we combine the tracking function $\delta$ with a function
for propositional transformation of formula sets guided by letters
from~$\Sigma_p$, embodied in the function~$\gamma$ defined next.

\begin{defn}[Propositional transformation]
  We define $\gamma:\Pow(\clos)\times \Sigma_p\to\Pow(\clos)$ on
  $\Gamma\subseteq \clos$, $(\chi,b)\in\Sigma_p$ by
  $\gamma(\Gamma,(\chi,b))=\Gamma$ if $\chi\notin\Gamma$, and
  \begin{equation*}
    \gamma(\Gamma,(\chi,b))= (\Gamma\setminus{\chi})\cup
       \begin{cases}
         \{\phi_b\} &  \chi=\phi_0\vee\phi_1\\
         \{\phi_0,\phi_1\} & \chi=\phi_0\wedge\phi_1\\
          \{\phi_0[\eta X.\,\phi_0/X]\} & \chi=\eta X.\,\phi_0
       \end{cases}
     \end{equation*}
     if $\chi\in\Gamma$. We extend $\gamma$ to words over $\Sigma_p$
     in the obvious way.
\end{defn}
\begin{expl}
  Take $\phi=\chi\vee\nu X.\,(\psi_1\wedge\psi_2)$ and
  $w=(\chi\vee\nu X.\,(\psi_1\wedge\psi_2),1),(\nu
  X.(\psi_1\wedge\psi_2),0),((\psi_1\wedge\psi_2)[\theta(X)/X],1)$. The letter
  $(\chi\vee\nu X.\,(\psi_1\wedge\psi_2),1)$ picks the right disjunct
  and $(\nu X.\,(\psi_1\wedge\psi_2),0)$ passes through the fixpoint
  operator $\nu X$ to reach $\psi_1\wedge\psi_2$; the letter
  $(\psi_1\wedge\psi_2,1)$ picks both conjuncts. Thus,
  \begin{equation*}
    \gamma(\{\phi\},w)=\{\psi_1[\theta(X)/X],\psi_2[\theta(X)/X]\}.
  \end{equation*}
\end{expl}

\section{Tableaux}\label{sec:tableaux}

As a stepping stone between neighbourhood models and satisfiability
games, we now introduce \emph{tableaux} which are built using a
variant of the standard \emph{tableau rules}~\cite{FriedmannLange13a}
(each consisting of one \emph{premise} and a possibly empty set of
\emph{conclusions}), where the \emph{modal rule}~$(\diamod{a})$
reflects Lemma~\ref{lem:rule}, taking into account the global
assumption~$\assume$:
\begin{align*}
  (\bot)\quad & \;\;\;\;\quad\quad\infrule{\Gamma,\bot}
                {}
  &  (\lightning)\quad & \quad\quad\infrule{\Gamma,p,\overline{p}}
                         {}
  &
    (\wedge)\quad & \;\quad\quad\infrule{\Gamma,\phi_1\wedge \phi_2}
                    {\Gamma,\phi_1,\phi_2}
  \\[5pt] (\vee) \quad & \quad\infrule{\Gamma,\phi_1\vee \phi_2}
                         {{\Gamma,\phi_1}\qquad {\Gamma,\phi_2}}
  &
    (\diamod{a}) \quad & \infrule{\Gamma,\diamod{a} \phi_1,\boxmod{a}\phi_2}
                         {\phi_1,\phi_2,\assume}
  &(\eta) \quad &\infrule{\Gamma,\eta X.\, \phi}
                  {\Gamma,\phi [\eta X.\, \phi/X]}
\end{align*}
(for $a\in \mathsf{A}$, $p\in \mathsf{P}$); we usually write rule
instances with premise $\Gamma$ and conclusion
$\Theta=\Gamma_1,\ldots,\Gamma_n$ ($n\le 2$) inline as
$(\Gamma/\Theta)$. Looking back at Section~\ref{sec:tracking}, we see
that letters in $\Sigma$ designate rule applications, e.g. the letter
$(\diamod{a}\phi_1,\boxmod{a}\phi_2)\in\Sigma_m$ indicates application
of~$(\diamod{a})$ to formulae $\diamod{a}\phi_1$, $\boxmod{a}\phi_2$.

\begin{definition}[Tableaux]
  Let $\mathsf{states}$ denote the set of \emph{(formal) states},
  i.e.\ sets $\Gamma\subseteq\clos$ such that $\bot\notin\Gamma$,
  $\{p,\overline p\}\not\subseteq\Gamma$ for all $p\in\mathsf{P}$, and
  such $\Gamma$ does not contain formulae that contain
  top-level occurrences of the operators $\wedge$, $\vee$, or $\eta X$.
  

  A \emph{pre-tableau} is a directed graph $(W,L)$, consisting of a
  finite set~$W$ of nodes labelled with subsets of $\clos$ by a
  labeling function $l:W\to\Pow(\clos)$, and of a relation
  $L\subseteq W\times W$ such that for all nodes~$v\in W$ 
  with label $l(v)=\Gamma\in\mathsf{states}$
  and all
  applications $(\Gamma/\Gamma_1,\ldots,\Gamma_n)$ of a tableau rule
  to $\Gamma$, there is an edge $(v,w)\in L$ such
  that $l(w)=\Gamma_i$ for some $1\leq i\leq n$.
  For nodes with label $l(v)=\Gamma\notin\mathsf{states}$, we require
  that there is exactly one node $w\in W$ such that $(v,w)\in L$;
  then we demand that there is an application $(\Gamma/\Gamma_1,\ldots,\Gamma_n)$ of a non-modal rule to $\Gamma$ such that
  $l(w)=\Gamma_i$ for some $1\leq i\leq n$.
  Finite or infinite words over $\Sigma$ then encode sequences of rule
  applications (and choices of conclusions for disjunctions). That is,
  given a starting node $v$, they encode \emph{branches} with root
  $v$, i.e.\ (finite or infinite) paths through $(W,L)$ that start
  at~$v$.  
  
  Deferrals are tracked along rule applications by means of the
  function $\delta$. E.g. for a deferral $\phi_0\vee\phi_1$, the
  letter $l=(\phi_0\vee\phi_1,0)$ identifies application of~$(\vee)$
  to $\{\phi_0\vee\phi_1\}$, and the choice of the left disjunct; then
  $\phi_0\vee\phi_1$ is tracked from a node with label
  $\Gamma\cup\{\phi_0\vee\phi_1\}$ to a successor node with label
  $\Gamma\cup\{\phi_0\}$ if $\phi_0\in\delta(\phi_0\vee\phi_1,l)$,
  that is, if~$\phi_0$ is a deferral.  A \emph{trace} (of $\phi_0$)
  along a branch with root~$v$ (whose label contains~$\phi_0$) encoded
  by a word~$w$ is a (finite or infinite) sequence
  $t=\phi_0,\phi_1\ldots$ of formulae such that
  $\phi_{i+1}\in\delta(\phi_i,w(i))$.  A \emph{tableau} is a finite
  pre-tableau in which all traces are finite, and a tableau is a
  \emph{tableau for $\target$} if some node label contains~$\target$.

  Given a tableau $(W,L)$ and a node $v\in W$, let $\mathsf{tab}(v)$
  (for \emph{tableau timeout}) denote the least number $m$ such that
  for all formulae $\phi$ in $l(v)$ and all branches of $(W,L)$ that
  are rooted at $v$, all traces of $\phi$ along the branch have length
  at most~$m$; such an~$m$ always exists by the definition of
  tableaux.
\end{definition}
\noindent To link models and tableaux, we next define an inductive
measure on unfolding of least fixpoint formulae in models.

\begin{defn}[Extension under timeouts]\label{defn:extto}
  Let $k$ be the greatest index of any least fixpoint variable
  in~$\psi$, and let $(W,N,I)$ be a finite neighbourhood model.  Then
  a \emph{timeout} is a vector $\overline{m}=(m_1,\ldots,m_k)$ of
  natural numbers $m_i\leq |W|$.  For $1\leq i\leq k$ such that
  $m_i>0$, we put
  $\overline{m}@i=(m_1,\ldots,m_{i-1},m_i-1,|W|,\ldots,|W|)$. Then
  $\overline{m}>_l\overline{m}@i$, where $>_l$ denotes lexicographic
  ordering.  For $\phi\in\clos$, we inductively define the
  \emph{extension $\sem{\phi}_{\overline{m}}$ under timeout
    $\overline{m}$} by $\sem{\phi}_{\overline{m}}=\sem{\phi}$
    for $\phi\notin\dfr$ and, for $\phi\in\dfr$,
    by
\begin{align*}
\sem{\psi_0\wedge\psi_1}_{\overline{m}} & = \sem{\psi_0}_{\overline{m}}\cap\sem{\psi_1}_{\overline{m}} \\
\sem{\psi_0\vee\psi_1}_{\overline{m}} & = \sem{\psi_0}_{\overline{m}}\cup\sem{\psi_1}_{\overline{m}}\\
\sem{\diamod{a}\psi}_{\overline{m}} & = \{w\in W\mid \exists S\in N(a,w).\,S\subseteq
\sem{\psi}_{\overline{m}}\}\\
\sem{\boxmod{a}\psi}_{\overline{m}} & = \{w\in W\mid 
\forall S\in N(a,w).\,S\cap\sem{\psi}_{\overline{m}}\neq \emptyset\}\\
\sem{\mu X.\psi}_{\overline{m}}&=\begin{cases}
\emptyset & m_{\mathsf{idx}(X)}=0\\
\sem{\psi[\mu X.\psi/X]}_{\overline{m}@\mathsf{idx}(X)} & m_{\mathsf{idx}(X)}>0
\end{cases}
\end{align*}
using the lexicographic ordering on $(\overline{m},\phi)$ as the
termination measure; crucially, unfolding least fixpoints reduces the
timeout. We extend this definition to sets of formulae by 
$\sem{\Psi}_{\overline{m}}=
\bigcap_{\phi\in\Psi} \sem{\phi}_{\overline{m}}$ for $\Psi\subseteq \cl$.

\end{defn}

\begin{lem}\label{lem:universal-timeouts}
In finite neighbourhood models, we have that 
for all $\psi\in\clos$ there is some timeout $\overline{m}$ such that
\begin{align*}
\sem{\psi}\subseteq
\sem{\psi}_{\overline{m}}.
\end{align*}
\end{lem}
\begin{proof}
 Let $W$ be the set of states.  Since
  $\sem{\psi}_{\overline{m}}$ is defined like $\sem{\psi}$ in all
  cases but one, we only need to consider the inductive case where
  $\psi=\mu X.\psi_0\in\clos$.  We show that for all
  \emph{subformulae}~$\phi$ of~$\psi_0$, for all timeout vectors
  $\overline{m}=(m_1,\ldots,m_k)$ such that $m_i=|W|$ for all~$i$ such
  that~$\phi$ does not contain a free variable with index at
  least~$i$, we have
\begin{equation}\label{eq:subform-to}
\sem{\phi}_{\sigma(\overline{m})}\subseteq\sem{\clf{\phi}}_{\overline{m}},
\end{equation}
where~$\sigma(\overline{m})$ maps each $Y\in\mathsf{FV}(\phi)$ to
$(\sem{\phi_1}^Y_{\sigma(\overline{m})})^{m_{\mathsf{idx}(Y)}}(\emptyset)$
where $\theta(Y)=\mu Y.\,\phi_1$ (this is a recursive definition of
$\sigma(\overline{m})$ since the value of $\sigma(\overline{m})$
for~$Y$ is overwritten in
$(\sem{\phi_1}^Y_{\sigma(\overline{m})})^{m_{\mathsf{idx}(Y)}}(\emptyset)$
so that this set depends only on values~$m_i$ such that
$i<\mathsf{idx}(Y)$). This shows that the claimed property
$\sem{\mu X.\,\psi_0}\subseteq \sem{\mu X.\,\psi_0}_{\overline{m}}$ holds
for $\overline{m}=(|W|,\ldots,|W|)$: By Kleene's fixpoint theorem,
\begin{align*}
\sem{\mu
  X.\psi_0}&=(\sem{\psi_0}^X)^{|W|}(\emptyset)=\sem{\psi_0}_{[X\mapsto(\sem{\psi_0}^X)^{|W|-1}(\emptyset)]} = \sem{\psi_0}_{\sigma(\overline
  m@\mathsf{idx}(X))}\\
&\subseteq
\sem{\clf{\psi_0}}_{\overline{m}@\mathsf{idx(X)}}= \sem{\mu
  X.\psi_0}_{\overline{m}},
  \end{align*}
   using~\eqref{eq:subform-to} in the
second-to-last step.

The proof of~\eqref{eq:subform-to} is by induction on $\phi$; we do
only the non-trivial cases. 
If $\clf{\phi}\notin\dfr$, then $\phi$ is closed (since $\psi_0$ does
not contain a free greatest fixpoint variable and since $\phi$ is not
in the scope of a greatest fixpoint operator within~$\psi_0$, that is,
no free greatest fixpoint variable is introduced during the induction).
Hence we have
$\sem{\phi}_{\sigma(\overline{m})}=\sem{\phi}$ and
$\sem{\clf{\phi}}_{\overline{m}}=\sem{\clf{\phi}}=\sem{\phi}$ so that we are
done.  If $\phi = Y$, then
$\sem{Y}_{\sigma(\overline{m})}=(\sem{\phi_1}^Y_{\sigma(\overline{m})})^{m_{\mathsf{idx}(Y)}}(\emptyset)$
where $\theta(Y)=\mu Y.\,\phi_1$.  If $m_{\mathsf{idx}(Y)}=0$, then
$(\sem{\phi_1}^Y_{\sigma(\overline{m})})^{m_{\mathsf{idx}(Y)}}(\emptyset)=\emptyset$
so there is nothing to show.  If $m_{\mathsf{idx}(Y)}>0$, then
\begin{align*}
(\sem{\phi_1}^Y_{\sigma(\overline{m})})^{m_{\mathsf{idx}(Y)}}(\emptyset)&=
\sem{\phi_1}^Y_{\sigma(\overline{m})}((\sem{\phi_1}^Y_{\sigma(\overline{m})})^{m_{\mathsf{idx}(Y)-1}}(\emptyset))\\
&=\sem{\phi_1}_{({\sigma(\overline{m})})[Y\mapsto (\sem{\phi_1}^Y_\sigma)^{m_{\mathsf{idx}(Y)}-1}(\emptyset)]}\\
&\subseteq
\sem{\clf{\phi_1}}_{\overline{m}@\mathsf{idx}(Y)}\\
&=
\sem{\clf{\phi_1[\mu Y.\,\phi_1/Y]}}_{\overline{m}@\mathsf{idx}(Y)}\\
&=
\sem{\clf{\mu Y.\,\phi_1}}_{\overline{m}}
=\sem{\clf{Y}}_{\overline{m}},
\end{align*}
where the inclusion is by the inductive hypothesis.
If $\phi=\mu Y.\,\phi_1$, then 
\begin{align*}
\sem{\mu Y.\,\phi_1}_{\sigma(\overline{m})}=(\sem{\phi_1}_{\sigma(\overline{m})}^Y)^{|W|}(\emptyset)
&=\sem{\phi_1}^Y_{\sigma(\overline{m})}((\sem{\phi_1}^Y_{\sigma(\overline{m})})^{|W|-1}(\emptyset))\\
&=\sem{\phi_1}_{({\sigma(\overline{m})})[Y\mapsto(\sem{\phi_1}^Y)^{|W|-1}(\emptyset)]}\\
&\subseteq\sem{\clf{\phi_1}}_{\overline{m}@\mathsf{idx}(Y)}\\
&=\sem{\clf{\phi_1[\theta(Y)/Y]}}_{\overline{m}@\mathsf{idx}(Y)}\\
&=\sem{\clf{\mu Y.\,\phi_1}}_{\overline{m}},
\end{align*}
where the first equality is by Kleene's fixpoint theorem and
the inclusion is by the inductive hypothesis since $m_{\mathsf{idx}(Y)}=|W|$
(hence $m@\mathsf{idx}(Y)_{\mathsf{idx}(Y)}=|W|-1$) by assumption as $\mu Y.\,\phi_1$ does not contain a free variable
with index at least $\mathsf{idx}(Y)$.\qed
\end{proof}

\noindent To check satisfiability it suffices to decide whether a tableau
exists:

\begin{theorem}\label{thm:sattab}
  There is a tableau for $\target$ if and only if $\target$ is
  $\assume$-satisfiable.
\end{theorem}
\begin{proof}
Let $\target$ be $\assume$-satisfiable. Then there is 
a finite neighbourhood model $(W,N,I)$ such that $W\subseteq\sem{\assume}$
and $W\cap\sem{\target}\neq\emptyset$ by Lemma~\ref{lem:fin-model}.
We define a tableau over the set 
$V=\{(x,\Psi,w)\in W\times\Pow(\clos)\times \Sigma_p^*\mid x\in\sem{\Psi},\mathsf{u}(\Psi)\leq 3n-|w|\}$,
where $\mathsf{u}(\Psi)$ denotes the sum of the numbers of unguarded operators in formulae from $\Psi$. Let $(x,\Psi,w)\in V$. For
$\Psi\notin\mathsf{states}$,
$\mathsf{u}(\Psi)\leq 3n-|w|$, we pick some
$\phi\in\Psi$, distinguishing
cases. If $\phi=\phi_0\wedge\phi_1$ or $\phi=\eta X.\,\phi_0$,
then we put $b=0$. If $\phi=\phi_0\vee\phi_1$, then let $\overline{m}$ be the least timeout such that $x\in\sem{\phi}_{\overline{m}}$ (such $\overline{m}$ exists by Lemma~\ref{lem:universal-timeouts}). Then there
is some $b'\in\{0,1\}$ such that $x\in\sem{\phi_{b'}}_{\overline{m}}$ and
we put $b=b'$. In any case, we put $l=(\phi,b)$ and add an edge from $(x,\Psi,w)$ to $(x,\gamma(\Psi,l),(w,l))$
to $L$, having $\mathsf{u}(\gamma(\Psi,l))=\mathsf{u}(\Psi)-1$
and hence $\mathsf{u}(\gamma(\Psi,l))\leq 3n-|w,l|$.
Since we are interested in the nodes that can be reached from
nodes $(x,\Psi,\epsilon)$ with $|\Psi|\leq 3$ and since each formula
in $\Psi$ contains at most $n$ unguarded operators, we indeed only
have to construct nodes $(x,\Psi',w)$ with $\mathsf{u}(\Psi)\leq 3n-|w|$
before reaching state labeled nodes. For
$\Psi\in\mathsf{states}$ and
$ \{\diamod{a}\phi_0,\boxmod{a}\phi_1\}\subseteq \Psi$,
let $\overline{m}$ be the least timeout such that $x\in\sem{\{\diamod{a}\phi_0,\boxmod{a}\phi_1\}}_{\overline{m}}$ (again, such $\overline{m}$ exists by Lemma~\ref{lem:universal-timeouts}). Then there is 
$S\in N(a,x)$ such that $S\subseteq\sem{\phi_0}_{\overline{m}}\cap\sem{\assume}$ and $S\cap\sem{\phi_1}_{\overline{m}}\neq \emptyset$.
Pick $y\in\sem{\phi_0}_{\overline{m}}\cap\sem{\phi_1}_{\overline{m}}\cap\sem{\assume}$ and add an edge from $(x,\Psi,w)$ to
$(y,\{\phi_0,\phi_1,\assume\},\epsilon)$ to $L$, having
$\mathsf{u}(\{\phi_0,\phi_1,\assume\})\leq 3n - |\epsilon|$.
Define the label $l(x,\Psi,w)$ of nodes $(x,\Psi,w)$ to 
be just $\Psi$. Then the structure $(V,L)$ indeed is a tableau for 
$\target$:
Since there is $z\in W\cap\sem{\target}$, we have 
$(z,\{\target\},\epsilon)\in V$, that is, $\target$ is
contained in the label of some node in $(V,L)$. The requirements
of tableaux for matching rule applications are satisfied
by construction of $L$. It remains to show that each trace in $(V,L)$
is finite. So let $\tau=\phi_0,\phi_1,\ldots$ be a trace of some formula
$\phi_0$ along some branch (encoded by a word $w$) that is rooted
at some node $v=(x,\Psi,w')\in V$ such that $\phi_0\in\Psi$.
Let $\overline{m}$ be the least timeout such that $x\in\sem{\phi_0}_{\overline{m}}$ (again, such $\overline{m}$ exists by Lemma~\ref{lem:universal-timeouts}). For each $i$, we have $\phi_{i+1}\in\delta(\phi_i,w(i))$ and there are 
$(x_i,\Psi_i,w_i)\in V$ and $\overline{m}_i$ such that $\phi_i\in\Psi_i$
and $x_i\in\sem{\phi_i}_{\overline{m}_i}$. We have chosen disjuncts
and modal successors in a minimal fashion,
 that is, in such a way that we always
have $\overline{m}_{i+1}\leq_l \overline{m}_{i}$. Since traces can be
infinite only if they contain infinitely many unfolding steps for
some least fixpoint, $\tau$ is finite or some least fixpoint is unfolded
in it infinitely often. In the former case, we are done. 
In the latter case, each unfolding of a least fixpoint 
by Definition~\ref{defn:extto} reduces some digit of the timeout 
so that we have an infinite decreasing chain
$\overline{m}_{i_1}>_l\overline{m}_{i_2}>_l\ldots$ with $i_{j+1}>i_j$ for all $j$, which is a contradiction to $<_l$ being a well-order.

For the converse direction, let $(V,L)$ be a tableau for $\target$
labeled by some function $l:V\to\Pow(\clos)$.
We construct a model $(W,N,I)$ over the set 
\begin{align*}
W=\{x\in V \mid l(x)\in\mathsf{states}\}.
\end{align*}
For $p\in\mathsf{P}$, put $I(p)=\{x\in W\mid p\in l(x)\}$;
since $(V,L)$ is a tableau, the rules $(\bot)$ and $(\lightning)$
do not match the label of any node, so we have $I(p)=W\setminus I(\overline{p})$, as required. 
Let $x\in W$ and $a\in\mathsf{A}$.
If $l(x)$ contains no $a$-box literal, then we put
$N(a,x)=\{\emptyset\}$. 
If $l(x)$ contains some $a$-box literal but no
$a$-diamond literal, then we put
$N(a,x)=\emptyset$. 
Otherwise, let the $a$-modalities
in $l(x)$ be exactly $\diamod{a}\chi_1,\ldots,\diamod{a}\chi_o,
\boxmod{a}\psi_1,\ldots,\boxmod{a}\psi_m$. We put
$N(a,x)=\{\{y_{1,1},\ldots,y_{1,m}\},\ldots,\{y_{o,1},\ldots,y_{o,m}\}\}$,
where, for $1\leq i\leq o$, $1\leq j\leq m$,
the state $y_{i,j}$ is picked minimally with respect to
$\mathsf{tab}$ among all nodes $z$ such that $(x,z)\in L$ and
$\{\chi_i,\psi_j,\assume\}= l(z)$; such $y_{i,j}$ with minimal
tableau timeout
always exists since $(V,L)$ is a tableau. It remains
to show that $(W,N,I)$ is a $\assume$-model for $\target$. 
We put
\begin{equation*}
\psem{\phi}=\{x\in W\mid l(x)\vdash_\mathsf{PL}\phi\}
\end{equation*}
  for $\phi\in\clos$, where
  $\vdash_{\mathsf{PL}}$ denotes \emph{propositional entailment}
  (modal literals $\boxmod{a}\phi$, $\diamod{a}\phi$ are
  regarded as propositional atoms and~$\eta X.\,\phi$ and 
  $\phi[\eta X.\,\phi/X]$ entail each other).
Since we have $W\subseteq\psem{\assume}$ and
$W\cap\psem{\target}\neq\emptyset$ by definition of $(W,N,I)$ and
tableaux, it
suffices to show that we have
$\psem{\phi}\subseteq\sem{\phi}$
for all $\phi\in\clos$.
The proof of this is by induction on $\phi$, using
coinduction in the case for greatest fixpoint formulae
and a further induction on tableau timeouts
in the case for least fixpoint formulae.

\qed
\end{proof}

\section{Satisfiability Games}\label{sec:games}

We now define a game characterizing $\assume$-satisfiability of
$\target$. Player $\mathsf{Eloise}$ tries to establish the existence
of a tableau for $\target$ using only polynomially many supporting
points for her reasoning. To this end, sequences of propositional
reasoning steps are contracted into single
$\mathsf{Eloise}$-moves. Crucially, the limited branching of monotone
modalities (and the ensuing limited need for tracking of deferrals)
has a restricting effect on the nondeterminism that the game needs to
take care of.
 
\begin{defn}[Satisfiability games]\label{defn:satgames}
  We put $U=\{\Psi\subseteq \clos\mid 1\leq|\Psi|\leq 2\}$ (note
  $|U|\leq n^2$) and
  $Q=\{\mathsf{Foc}\subseteq \dfr\mid |\mathsf{Foc}|\leq 2\}$.  The
  \emph{$\assume$-satisfiability game for~$\target$} is the B\"uchi
  game $G=(V,E,v_0,F)$ with set $V=V_\exists\cup V_\forall$ of nodes
  (with the union made disjoint by markers omitted in the notation)
  where
  $V_\exists=\{(\Psi,\mathsf{Foc})\in U\times Q\mid
  \mathsf{Foc}\subseteq\Psi$\} and
  $V_\forall=\mathsf{states}\times Q$, with initial node
  $v_0=(\{\target\},\emptyset)\in V_\exists$, and with set
  $F=\{(\Psi,\mathsf{Foc})\in V_\exists\mid \mathsf{Foc}=\emptyset\}$
  of accepting nodes. 
  The set $E$ of moves is defined by 
\begin{align*}
E(\Psi,\mathsf{Foc})& =
\{(\gamma(\Psi\cup\{\assume\},w),\delta(\mathsf{Foc},w))\in \mathsf{states}\times Q\mid w\in(\Sigma_p)^*, |w|\leq 3n\}\\
E(\Gamma,\mathsf{Foc})& =
\{\,(\{\phi_0,\phi_1\},\mathsf{Foc}')\in \,  U\times Q \mid \, 
\{\diamod{a}\phi_0,\boxmod{a}\phi_1\}\subseteq \Gamma,\\
&
\hspace{6em}\text{if }\mathsf{Foc}\neq \emptyset \text{, then }\mathsf{Foc}'=\delta(\mathsf{Foc},(\diamod{a}\phi_0,\boxmod{a}\phi_1)),\\
&\hspace{6em}\text{if }\mathsf{Foc}=\emptyset \text{, then }\mathsf{Foc}'=\{\phi_0,\phi_1\}\cap\dfr\,
\}
\end{align*}
for $(\Psi,\mathsf{Foc})\in V_\exists$,
$(\Gamma,\mathsf{Foc})\in V_\forall$.
\end{defn}

\noindent Thus, \textsf{Eloise} steers the propositional evolution of
formula sets into formal states, keeping track of the focussed
formulae, while \textsf{Abelard} picks an application of
Lemma~\ref{lem:rule}, and resets the focus set after it is
\emph{finished}, i.e.\ becomes~$\emptyset$; \textsf{Eloise} wins plays
in which the focus set is finished infinitely often.
\begin{rem}
  It is crucial that while the game has exponentially many
  $\mathsf{Abelard}$-nodes, there are only polynomially many
  $\mathsf{Eloise}$-nodes. In fact, all
  $\mathsf{Eloise}$-nodes $(\Psi,\mathsf{Foc})$ have
  $\mathsf{Foc}\subseteq\Psi$, so the game to
  has at most $4|U|\leq 4n^2$ $\mathsf{Eloise}$-nodes.
\end{rem}


\noindent Next we prove the correctness of our satisfiability games.

\begin{theorem}\label{thm:tabgame}
There is a tableau for $\target$ if and only if $\mathsf{Eloise}$ wins $G$.
\end{theorem}
\begin{proof}
Let $s$ be a winning strategy for $\mathsf{Eloise}$ in $G$
with which she wins every node in her winning region $\mathsf{win}_\exists$.
For $v=(\Psi,\mathsf{Foc})\in \mathsf{win}_\exists$,
we let $w_{s(v)}$ denote a fixed propositional word such that
$s(\Psi,\mathsf{Foc})=(\gamma(\Psi,w_{s(v)}),\delta(\mathsf{Foc},w_{s(v)}))$, that is, a witness word for the move of $\mathsf{Eloise}$ that
$s$ prescribes at $v$.
We construct a tableau $(W,L)$ over the set
\begin{align*}
W=\{(\gamma(\Psi\cup\{\assume\},w'),\delta(\mathsf{Foc},w'))\in&\,\Pow(\clos)\times\Pow(\dfr)\mid \\
&(\Psi,\mathsf{Foc})\in\mathsf{win}_\exists, w'\text{ is a prefix of } w_{s(\Psi,\mathsf{Foc})}\}.
\end{align*}
We define the label $l(\Phi,\mathsf{Foc}')$ of nodes from $(\Phi,\mathsf{Foc}')\in W$ to be just $\Phi$.
Let $(\Phi,\mathsf{Foc}')=(\gamma(\Psi\cup\{\assume\},w'),\delta(\mathsf{Foc},w'))\in W$. If $w'$ is a proper prefix of $w_{s(\Psi,\mathsf{Foc})}$, then
we have $\Phi\notin\mathsf{states}$;
let $l\in\Sigma_p$ be \emph{the} letter such that $(w',l)$ is
a prefix of $w_{s(\Psi,\mathsf{Foc})}$ and
add the pair $((\Phi,\mathsf{Foc}'),(\gamma(\Phi,l),\delta(\mathsf{Foc}',l)))$
to $L$.  If $w'=w_{s(\Psi,\mathsf{Foc})}$, then we have
$\Phi\in\mathsf{states}$.
For 
$\{\diamod{a}\phi,\boxmod{a}\chi\}\subseteq\Phi$, we distinguish
cases. If $\mathsf{Foc}'=\emptyset$, then put
$\mathsf{Foc}''=\{\phi,\chi\}\cap\dfr$; otherwise, put
$\mathsf{Foc}''=\delta(\mathsf{Foc}',(\diamod{a}\phi,\boxmod{a}\chi))$.
Then add the pair
$((\Phi,\mathsf{Foc}'),(\{\phi,\chi,\assume\},\mathsf{Foc}''))$
to $L$, having
$(\{\phi,\chi\},\mathsf{Foc}'')\in\mathsf{win}_\exists$.  It remains
to show that all traces in $(W,L)$ are finite.  So let
$\tau=\phi_0,\phi_1,\ldots$ be a trace along some branch (encoded by a
word $w$) that is rooted at some node
$(\Phi,\mathsf{Foc}')\in W$ such that
$\phi_0\in \Phi$. By construction, this branch gives rise to an $s$-play
$(\Psi_0,\mathsf{Foc}_0),(\Gamma_0,\mathsf{Foc}'_0),
(\Psi_1,\mathsf{Foc}_1),(\Gamma_1,\mathsf{Foc}'_1),\ldots$ that starts
at $(\Psi_0,\mathsf{Foc}_0)=(\Psi,\mathsf{Foc})$.  
Let $i$ be the least position such that $\mathsf{Foc}'_i=\emptyset$ ($i$
exists because~$s$ is a winning strategy). Since the~$\phi_j$ are
tracked along rule applications, we have $\phi_i\in\Gamma_i$; hence
$\phi_{i+1}\in\mathsf{Foc}_{i+1}$. Let $i'$ be the least position greater
than~$i$ such that $\mathsf{Foc}'_{i'}=\emptyset$ (again,~$i'$ exists
because $s$ is a winning strategy).  Between
$(\Psi_{i+1},\mathsf{Foc}_{i+1})$ and
$(\Gamma_{i'},\mathsf{Foc}'_{i'})$, all formulae from
$\mathsf{Foc}_{i+1}$ (including $\phi_{i+1}$) are transformed to a
non-deferral by the formula manipulations encoded in~$w$.  In
particular, the trace $\tau$ ends between
$\mathsf{node}(\Psi_{i},\mathsf{Foc}_{i})$ and
$\mathsf{node}(\Psi_{i'},\mathsf{Foc}_{i'})$, and hence is finite.

For the converse direction, let $(W,L)$ be tableau for $\target$, 
labeled with
$l:W\to\Pow(\clos)$. We extract a strategy $s$ for
$\mathsf{Eloise}$ in $G$. 
A game node $(\Psi,\mathsf{Foc})$ is \emph{realized} if
there is $v\in W$ such that $\Psi\subseteq l(v)$; then we
say that $v$ \emph{realizes} the game node.
For all realized game nodes $(\Psi,\mathsf{Foc})$, we pick a realizing
tableau node $v(\Psi,\mathsf{Foc})$ that is minimal with respect to 
$\mathsf{tab}$ among the tableau nodes that realize $(\Psi,\mathsf{Foc})$.
Then we construct a propositional word $w=l_0,l_1,\ldots$ as follows,
starting with $\Psi_0=\Psi\cup\{\assume\}$ and $v_0=v(\Psi,\mathsf{Foc})$.
For $i\geq 0$, pick some non-modal letter $l_i=(\phi,b)$ such that
$\phi\in\Psi_i$, $b\in\{0,1\}$
and such that $l(v_{i+1})=\gamma(\Psi_i,l_i)$
where $v_{i+1}\in W$ is \emph{the}
node such that $(v_i,v_{i+1})\in L$.
Such a letter $l_i$ exists since $(W,L)$ is a tableau.
By guardedness of fixpoint variables, this process will eventually 
terminate with a word $w=l_0,l_1,\ldots,l_m$ such that $m\leq 3n$,
since $\Psi_0$ contains at most three formulae and each formula
contains at most $n$ unguarded operators. Put
$s(\Psi,\mathsf{Foc})=(\gamma(\Psi\cup\{\assume\},w),\delta(\mathsf{Foc},w))$,
having $\gamma(\Psi\cup\{\assume\},w)= l(v_m)$.
It remains to show that $s$ is a winning strategy. So let 
$\tau=(\Psi_0,\mathsf{Foc}_0),(\Gamma_0,\mathsf{Foc}'_0),(\Psi_1,\mathsf{Foc}_1),(\Gamma_1,\mathsf{Foc}'_1),\ldots$ be an $s$-play, where $\Psi_0=\{\target\}$ and
$\mathsf{Foc}_0=\emptyset$. It suffices to show that for all $i$ such that
$\mathsf{Foc}_i\neq\emptyset$, there
is $j\geq i$ such that $\mathsf{Foc}_j=\emptyset$.
So let $\mathsf{Foc}_i\neq\emptyset$ and
let $w_i$ be the word that is constructed in the play
from $(\Psi_i,\mathsf{Foc}_i)$ on.
Since $(W,L)$ is a tableau and since
$s$ has been constructed using realizing nodes with minimal tableau timeouts, all traces of formulae from
$\mathsf{Foc}_i$ along the branch that is encoded by $w_i$
are finite. Let $j$ be the least
number such that all such traces have ended after
$2j$ further moves from $(\Psi_i,\mathsf{Foc}_i)$.
Then we have $\mathsf{Foc}'_{i+j}=\emptyset$, as required.
\qed
\end{proof}

\begin{cor}\label{cor:size}
  Every satisfiable formula of size $n$ in the alternation-free
 monotone $\mu$-calculus with the universal modality has a model of size at most
  $4n^2$.
\end{cor}

\begin{cor}\label{thm:mon-mu-np}
The satisfiability checking problem for the alternation-free 
monotone $\mu$-calculus  with the universal modality is in \NP{} (hence \NP-complete).
\end{cor}
\begin{proof}
Guess a winning strategy $s$ for $\mathsf{Eloise}$ in
$G$ 
 and verify that $s$ is a winning strategy. Verification can be done in polynomial
time since the structure obtained from 
$G$ by imposing $s$ is of polynomial size and since the
admissibility of single moves can be checked in polynomial time.\qed
\end{proof}
\noindent By the translations recalled in
Example~\ref{expl:game-cpdl}, we obtain moreover
\begin{cor}\label{cor:game-cpdl}
  Satisfiability-checking in concurrent propositional dynamic logic
  CPDL and in the alternation-free fragment of game logic is in \NP{} (hence \NP-complete).
\end{cor}

\section{Conclusion}

\noindent We have shown that satisfiability checking in the
alternation-free fragment of the monotone $\mu$-calculus with the
universal modality is only \NP-complete, even when formula size is
measured as the cardinality of the closure. Via straightforward
translations (which have only quadratic blow-up under the mentioned
measure of formula size), it follows that both concurrent
propositional dynamic logic (CPDL) and the alternation-free fragment
of game logic are also only \NP-complete under their original
semantics, i.e.\ with atomic programs interpreted as neighbourhood
structures (they become \ExpTime-complete when atomic programs are
interpreted as relations). We leave as an open problem whether the
upper bound \NP{} extends to the full monotone $\mu$-calculus, for
which the best known upper bound thus remains \ExpTime{}, by results
on the coalgebraic $\mu$-calculus~\cite{CirsteaEA11a}, or
alternatively by the translation into the relational $\mu$-calculus
that we give in the proof of the finite model property
(Lemma~\ref{lem:fin-model}).

\bibliographystyle{myabbrv}
\bibliography{coalgml}


\newpage
\appendix

\section{Omitted lemmas and proofs}

\subsection*{Full Proof of Lemma~\ref{lem:fin-model}}

\noindent We write relational models in the form
$C=(W,(R_a)_{a\in\mathsf{R}},I)$ where~$\mathsf{R}$ is a set of
relation symbols,~$W$ is the set of worlds and~$I$ interprets
propositional atoms like for neighbourhood models, and for each
relation symbol~$a\in\mathsf{R}$, $R_a\subseteq W\times W$ is a
transition relation. We interpret the relational $\mu$-calculus over
relational models in the standard fashion.

We define a translation~$t$ from monotone $\mu$-calculus formulae into
relational $\mu$-calculus formulae built using as relation symbols the
given actions in~$\mathsf{A}$ and an additional fresh relation
symbol~$e$. The translation is inductively defined by
\vspace{-5pt}
\begin{equation*}
  t(\boxmod{a}\psi)  = \boxmod{a}\diamod{e}\psi
  \qquad t(\diamod{a}\psi)  = \diamod{a}\boxmod{e}\psi
\vspace{-5pt}
\end{equation*}
and commutation with all other constructs. We claim that the following
are equivalent for formulae $\psi,\phi$ in the monotone
$\mu$-calculus.
\begin{enumerate}[label=(\emph{\alph*})]
\item\label{item:nbhd-inf} $\psi$ is $\phi$-satisfiable (over a neighbourhood model).
\item\label{item:rel-inf} $t(\psi)\land \mathsf{\boxtimes}\, t(\phi)$ is satisfiable
  (over a relational model).
\item\label{item:rel-fin-ag} $t(\psi)\land \mathsf{\boxtimes}\, t(\phi)$ is satisfiable over a 
  finite relational model.
\item\label{item:rel-fin} $ t(\psi)$ is satisfiable over a finite
  relational model in which every state satisfies~$t(\phi)$.
\item\label{item:nbhd-fin} $\psi$ is $\phi$-satisfiable over a finite
  neighbourhood model.
\end{enumerate}

Here, $\mathsf{\boxtimes}\,\phi$ abbreviates the formula
\vspace{-7pt}
\begin{equation*}
  \nu X.\,\phi\land\Box X
\vspace{-7pt}
\end{equation*}
where $\Box X$, in turn, is the conjunction of all formulae
$\boxmod{a} X$ for $a\in\mathsf{A}\cup\{e\}$; that is, $\boxtimes$ is
the modality `in all reachable states' mentioned in
Remark~\ref{rem:submodels}.

In the claimed equivalence,
\ref{item:nbhd-fin}$\implies$\ref{item:nbhd-inf} is trivial, and
\ref{item:rel-inf}$\implies$\ref{item:rel-fin-ag} is by the well-known
finite model property of the relational $\mu$-calculus. The
implication \ref{item:rel-fin-ag}$\implies$\ref{item:rel-fin} is by
restricting to the reachable part of the model.  The remaining
implications are proved as follows.

\ref{item:nbhd-inf}$\implies$\ref{item:rel-inf}: Let $F=(W,N,I)$ be a
neighbourhood model. We construct a relational model $C=(W',(R_a),I')$
as follows.
\begin{itemize}
\item We take $W'$ to be the union (w.l.o.g.\ disjoint) of~$W$ and
  $\Pow W$.
\item For $a\in\mathsf{A}$, we put
  $R_a=\{(x,A)\in W\times\Pow W\mid A\in N(a,x)\}$. For the fresh
  relation symbol~$e$, we put
  $R_e=\{(A,x)\in\Pow W\times W\mid x\in A\}$.
\item For $p\in \mathsf{P}$, we put $I'(p)=I(p)$.
\end{itemize}
We indicate the respective semantics in~$F$ and~$C$ by superscripts,
and show by induction over monotone $\mu$-calculus formulae~$\psi$
that
\vspace{-5pt}
\begin{equation*}
  \Sem{\psi}^F_\sigma\subseteq \Sem{t(\psi)}^C_{\sigma'}
\vspace{-5pt}
\end{equation*}
for valuations~$\sigma\colon V\to 2^W$, $\sigma'\colon V\to 2^{W'}$
such that $\sigma(Y)\subseteq\sigma'(Y)$ for all $Y\in V$; the claimed
implication \ref{item:nbhd-inf}$\implies$\ref{item:rel-inf} is then
immediate. The cases for Boolean operators, propositional atoms, and
variables are trivial; the modal cases are by noting that their
translations just reflect the definition of the semantics of the
monotone modalities. We do the fixpoint cases. For the least fixpoint
case, we show that
$\Sem{t(\mu X.\,\psi)}^C_{\sigma'}\cap W=\Sem{\mu X.\,t(\psi)}^C_{\sigma'}\cap W$ is a
prefixpoint of the function~$\Sem{\psi}^X_\sigma$ defining
$\Sem{\mu X.\,\psi}^F_\sigma$: We have
\begin{align*}
  & \sem{\psi}^X_\sigma(\Sem{\mu
    X.\,t(\psi)}^C_{\sigma'}\cap W) \\
  &=\sem{\psi}^F_{\sigma[X\mapsto \Sem{\mu
    X.\,t(\psi)}^C_{\sigma'}]\cap W}\by{definition}\\
  & \subseteq \sem{t(\psi)}^C_{\sigma'[X\mapsto \Sem{\mu
    X.\,t(\psi)}_{\sigma'}]} \by{IH}\\
  & =\Sem{\mu  X.\, t(\psi)}_{\sigma'}\by{fixpoint}.
\end{align*}
For the greatest fixpoint case, we show that
$\Sem{\nu X.\,\psi}_\sigma$ is a postfixpoint of the function
$\Sem{t(\psi)}^X_{\sigma'}$ defining
$\Sem{t(\nu X.\,\psi)}_{\sigma'}=\Sem{\nu X.\,(t(\psi))}_{\sigma'}$:
\begin{align*}
  & \Sem{\nu X.\,\psi}_\sigma \\
  & = \Sem{\psi}_{\sigma[X\mapsto \Sem{\nu X.\,\psi}_\sigma]} \by{fixpoint}\\
  & \subseteq \Sem{t(\psi)}_{\sigma'[X\mapsto \Sem{\nu X.\,\psi}_\sigma]}
    \by{IH}\\
  & = \Sem{t(\psi)}_{\sigma'}^X(\Sem{\nu X.\,\psi}_\sigma)\by{definition}.
\end{align*}

\ref{item:rel-fin}$\implies$\ref{item:nbhd-fin}: Let $C=(W,(R_a),I)$
be finite relational model. We define a neighbourhood model $F=(W,N,I)$
by
\begin{equation*}
  N(a,w)=\{\{w'\in W\mid (n,w')\in R_e\} \mid (w,n)\in R_a\}.
\end{equation*}
Similarly as above, we indicate the respective semantics in~$F$ or~$C$
by superscripts, and show by induction on monotone $\mu$-calculus
formulae~$\psi$ that
\begin{equation*}
  \Sem{\psi}_{\sigma}^F=\Sem{t(\psi)}_\sigma^C
\end{equation*}
for  valuations $\sigma\colon V\to 2^W$. The claimed
implication \ref{item:rel-fin}$\implies$\ref{item:nbhd-fin} is then
immediate.

The cases for Boolean operators, propositional atoms, variables, and
indeed for fixpoints are trivial; e.g.\ in the inductive step for
least fixpoints, just note that $\Sem{\mu X.\,\psi}^F_\sigma$ and
$\Sem{t(\mu X.\,\psi)}^C_\sigma=\Sem{\mu X.\,t(\psi)}^C_\sigma$ are,
by induction, least fixpoints of the same function on~$\Pow W$. The
modal cases are as follows.

$\boxmod{a}\psi$: We have
\begin{align*}
  & \Sem{\boxmod{a}\psi}^F_\sigma \\
  & = \{w\mid \forall S\in N(a,x).\,\exists w'\in S.\,w'\in\Sem{\psi}^F_\sigma\}
  \by{semantics}\\
  & = \{w\mid \forall S\in N(a,x).\,\exists w'\in S.\,w'\in\Sem{t(\psi)}^C_\sigma\}
    \by{IH}\\
  & = \{w\mid \forall n.\,(w,n)\in R_a\to\exists w'.\,(n,w')\in R_e\land w'\in\Sem{t(\psi)}^C_\sigma\}
    \by{construction}\\
  & = \Sem{\boxmod{a}\diamod{e}t(\psi)}^C_\sigma\by{semantics}\\
  & = \Sem{t(\boxmod{a}\psi)}^C_\sigma\by{definition}.
\end{align*}

$\diamod{a}\psi$: Dual to the previous case. \qed

\vspace{10pt}

\noindent We now carry out the
\textbf{full proof of the claim $\psem{\phi}\subseteq\sem{\phi}$ from the
proof of Theorem~\ref{thm:sattab}}.
The proof is by induction over $\phi$, 
where the Boolean and the modal
cases are straightforward. E.g.  let
$x\in\psem{\langle a\rangle\phi}$. Then we have
$\langle a\rangle\phi\in l(x)$.  If $l(x)$ contains no $a$-box formula,
then we have $N(a,x)=\{\emptyset\}$ so that we have
$x\in\sem{\langle a\rangle \phi}$ since we have $\emptyset\in N(a,x)$
and, trivially, for all $y\in\emptyset$, $y\in\sem{\phi}$. If $l(x)$
contains some $a$-box formula, then there is, by definition of $N$, 
some $S\in N(a,x)$ such that for all $y\in S$, we have
$y\in\psem{\phi}$. By the inductive hypothesis, we have
$y\in\sem{\phi}$ for all $y\in S$ so that
$x\in\sem{\langle a\rangle\phi}$, as required. For the fixpoint cases,
we proceed as follows:
\begin{itemize}[wide]
\item If $\phi=\nu X.\,\chi$, then we
have to show
\begin{equation*}
\psem{\nu X.\,\chi}\subseteq\sem{\nu X.\,\chi}=\bigcup\{Y\subseteq Z\mid
Y\subseteq \sem{\chi}^X(Y)\},
\end{equation*} 
which follows if $\psem{\nu X.\,\chi}$ is a postfixpoint of
$\sem{\chi}^X$, that is, if
$\psem{\nu X.\,\chi}=\psem{\chi}\subseteq \sem{\chi}^X(\psem{\nu
  X.\,\chi})=\sem{\chi}_{[X\mapsto\psem{\nu X.\,\chi}]}$.  We show more generally that for all
$\phi\in\mathsf{sub}(\chi)$ and all $\sigma$ such that for all
$Y\in\mathsf{FV}(\phi)$, $\sigma(Y)=\psem{\clf{Y}}$, we have
$\psem{\clf{\phi}}\subseteq \sem{\phi}_\sigma$.  We proceed by induction
over $\phi$.  If $\phi$ is closed, then $\phi\in\clos$, $\clf{\phi}=\phi$
and we are done by the outer
inductive hypothesis. The Boolean and modal cases are again
straightforward.  If $\phi=X$, then we have
$\sem{X}_\sigma=\sigma(X)=\psem{\clf{X}}$ so that
$\psem{\clf{X}}\subseteq \sem{X}_\sigma$, as required. The remaining case is
that $\phi$ is a fixpoint literal; since~$\phi$ is not closed,~$\phi$
is, by alternation-freeness, a greatest fixpoint
$\phi=\nu Y.\,\phi_1$. We have to show
$\psem{\clf{\nu Y.\,\phi_1}}=\psem{\clf{\phi_1}}\subseteq \sem{\nu
  Y.\,\phi_1}_\sigma$; by coinduction, it suffices to show that
$\psem{\clf{\phi_1}}\subseteq \sem{\phi_1}^Y_\sigma(\psem{\clf{\phi_1}})$. Since
$\sem{\phi_1}^Y_\sigma(\psem{\clf{\phi_1}})=\sem{\phi_1}^Y_\sigma(\psem{\clf{\nu
  Y.\,\phi_1}})=\sem{\phi_1}^Y_\sigma(\psem{\clf{Y}})=
\sem{\phi_1}_{\sigma[Y\mapsto(\psem{\clf{Y}})]}$, this follows from
the inner inductive hypothesis.
\item  If $\phi=\mu X.\,\chi$, then we
have to show
\begin{equation*}
\psem{\mu X.\,\chi}=\psem{\clf{\chi}}\subseteq\sem{\mu X.\,\chi}.
\end{equation*}
For $m\in\mathbb{N}$, let $\mathsf{tab}(m)$ denote the set of
states $x\in W$ such that $\mathsf{tab}(x)=m$.
It suffices to
show that for all subformulae $\phi$ of $\chi$ and all $m\in\mathbb{N}$, 
we have
\begin{equation*}
\psem{\clf{\phi}}\cap \mathsf{tab}(m)\subseteq\sem{\clf{\phi}}.
\end{equation*}
So let $x=(\Psi,\mathsf{Foc})\in \psem{\clf{\phi}}\cap \mathsf{tab}(m)$. We distinguish cases.
\begin{itemize}[wide]
\item [a)] If $\mathsf{Foc}\vdash_{\mathsf{PL}}\clf{\phi}$, then we
 proceed
  by lexicographic induction over $(m,\mathsf{u}(\phi),\mathsf{len}(\phi))$, where $\mathsf{u}(\phi)$ denotes the
  number of distinct fixpoint variables $X$ for which there is a
  formula that can be obtained from $\phi$ by repeatedly
  replacing fixpoint variables $Y$ with $\theta(Y)$ and in which
  $X$ is \emph{unguarded}. If $\phi$ is
  closed, then we are done by the outer induction hypothesis. 
  The Boolean cases are again straightforward.
  E.g. if $\phi=\psi_0\wedge\psi_1$, then we have
\begin{align*}  
  \psem{\clf{\phi}}\cap \mathsf{tab}(m)&=\psem{\clf{\psi_0}}\cap\psem{\clf{\psi_1}}\cap\mathsf{tab}(m)
\end{align*}  
   and $\sem{\clf{\phi}}=\sem{\clf{\psi_0}}\cap\sem{\clf{\psi_1}}$.
  We have 
  \begin{align*}
  \psem{\clf{\psi_0}}\cap\psem{\clf{\psi_1}}\cap \mathsf{tab}(m)\subseteq \psem{\clf{\psi_i}}\cap \mathsf{tab}(m)\subseteq\sem{\clf{\psi_i}}
  \end{align*}
   for $i\in\{0,1\}$,
  where the second inclusion holds by the inductive hypothesis
   since $(m,\mathsf{u}(\phi),\mathsf{len}(\phi))\geq_l(m,\mathsf{u}(\psi_i),\mathsf{len}(\psi_i))$.
   If $\chi=\diamod{a}\psi_0$, then we have $\clf{\chi}=\diamod{a}(\clf{\psi_0})$ and
   $\diamod{a}(\clf{\psi_0})\in l(x)$. If $l(x)$ contains
   no box formula, then we have $N(a,x)=\{\emptyset\}$ and
   hence $x\in\sem{\diamod{a}(\clf{\psi_0})}$, as required.
   If $l(x)$ contains some box formula, then there is, by definition
   of $N$, some $S\in N(a,x)$ such that for all
   $y\in S$, $y\in\psem{\clf{\psi_0}}$. Also, $y\in\mathsf{fto}(m-1)$ since
 we have taken at least one step in the tableau
 to get from $x$ to $y$. We have 
\begin{align*}
(m,\mathsf{u}(\diamod{a}\psi_0),\mathsf{len}(\diamod{a}\psi_0))>_l(m-1,\mathsf{u}(\psi_0),\mathsf{len}(\psi_0)),
\end{align*}
 even though possibly $\mathsf{u}(\diamod{a}\psi_0)<\mathsf{u}(\psi_0)$.
 For all $y\in S$, we have $y\in\sem{\clf{\psi_0}}$ by the inductive hypothesis. Hence  $x\in\sem{\diamod{a}(\clf{\psi_0})}$, as required.
 If $\phi=X$ with
  $\theta(X)=\mu X.\,\psi_1$, then
\begin{align*}  
 \psem{\clf{X}}\cap \mathsf{tab}(m)=\psem{\clf{\psi_1}}\cap\mathsf{tab}(m)
 \end{align*}
  and
  $\sem{\clf{X}}=\sem{\clf{\psi_1}}$. The fixpoint variable $X$ is
  unguarded in 
  the formula $X$ but it is not possible to (repeatedly) 
  replace fixpoint variables $Y$ in $\psi_1$ with $\theta(Y)$ 
  in such a way
  that $X$ becomes unguarded in the resulting formula.
  Thus we have $\mathsf{u}(X)=\mathsf{u}(\psi_1)+1$ and 
  $(m,\mathsf{u}(X),\mathsf{len}(X))>_l(m,\mathsf{u}(\psi_1),\mathsf{len}(\psi_1))$
  so that we are done by the inner inductive hypothesis.
  If $\phi=\mu Y.\,\phi_1$, then  
  \begin{align*}
  \psem{\clf{\mu
    Y.\,\psi_1}}\cap\mathsf{tab}(m)=\psem{\clf{\psi_1}}\cap\mathsf{tab}(m) 
    \end{align*}
     and  $\sem{\clf{\mu Y.\,\psi_1}}=\sem{\clf{\psi_1}}$. Since
     \begin{align*}
  (m,\mathsf{u}(\mu Y.\,\psi_1),\mathsf{len}(\mu Y.\,\psi_1))>_l
  (m,\mathsf{u}(\psi_1),\mathsf{len}(\psi_1)),
  \end{align*}
    we are done by the inner inductive hypothesis.
    Eventually, we reach a closed formula or the case where $m=0$,
    that is, a state where all traces have ended.
    Since traces end when the traced formula is
    transformed to a non-deferral 
    (which means that $\phi$ in the above induction becomes 
    a closed formula),
    $\phi$ is closed in the case that $m=0$. As mentioned above, the 
    outer inductive hypothesis finishes the case for closed formulae.

\item [b)]If $\mathsf{Foc}\not\vdash_{\mathsf{PL}}\clf{\phi}$, then we again proceed by induction over $(m,\mathsf{u}(\phi),\mathsf{len}(\phi))$. The inductive proof is identical to the previous item, with the exception of the base case with $m=0$. Then
we have reached a state with empty focus set and have
$\mathsf{Foc}'\vdash_{\mathsf{PL}}\phi'$ after next modal step which is a refocussing step. Then we proceed as in the previous item.
\end{itemize}

\end{itemize}
\qed

\subsection*{Proof of Corollary~\ref{cor:size}}

Let $\phi$ be a satisfiable formula in the alternation-free monotone 
$\mu$-calculus with the universal modality and put $n=|\phi|$.
By Theorem~\ref{thm:sattab} and Theorem~\ref{thm:tabgame},
$\mathsf{Eloise}$ wins the satisfiability game $G$. Then there
is a winning strategy $s$ for $\mathsf{Eloise}$. Using this strategy
in the tableau construction from the proof of Theorem~\ref{thm:tabgame},
we obtain a tableau which contains one formal state per 
$\mathsf{Eloise}$-node in the winning region $\mathsf{win}_\exists$ in $G$,
that is, there are at most $4n^2$ formal states.
Using the model construction from the proof of Theorem~\ref{thm:sattab},
we obtain a model that is built over the formal states of the tableau,
obtaining the claimed bound.
\qed

\end{document}